\documentclass[11pt]{article}
\usepackage{amsmath,amssymb,amsfonts,amsthm,array,hhline,setspace,graphicx,textcomp,overpic}
\usepackage{comment}
\usepackage[usenames]{color}

\specialcomment{genescomments}{}{}

\includecomment{genescomments}

\title{Asymptotic two-soliton solutions solutions in the Fermi-Pasta-Ulam model}
\author{Aaron Hoffman and C.E. Wayne \\ 
Boston University \\
Department of Mathematics and Statistics \\
111 Cummington St. \\
Boston, MA 02135}
\date{}
\newcommand{\eps}{\varepsilon}
\newcommand{\be}{\begin{equation}}
\newcommand{\ee}{\end{equation}}
\newcommand{\bea}{\begin{eqnarray}}
\newcommand{\eea}{\end{eqnarray}}

\newcommand{\ba}{\begin{array}}
\newcommand{\ea}{\end{array}}
\newcommand{\R}{\mathbb{R}}
\newcommand{\Z}{\mathbb{Z}}

\newtheorem{thm}{Theorem}[section]
\newtheorem{lemma}[thm]{Lemma}

\newtheorem*{rmk}{Remark}

\begin{document}
\maketitle
\abstract{We prove the existence of asymptotic two-soliton states in the Fermi-Pasta-Ulam model with general interaction potential.  That is, we exhibit solutions whose difference in $\ell^2$ from the linear superposition of two solitary waves goes to zero as time goes to infinity.}

\section{Introduction}
In a series of recent works, (e.g. \cite{martel:2001,martel:2002,martel:2005,martel:2006}) Martel, Merle, and collaborators have studied solutions of the generalized Korteweg DeVries (gKdV) equation 
\begin{equation}
 u_t = (u_{xx} + u^p)_x; \qquad x \in \R; \qquad p = 2,3,4,5. 
\label{eq:gKdV}
\end{equation}
It is a long standing observation that gKdV admits a family of solitary waves $u(x,t) = u_c(x - ct)$ in which larger waves move faster than smaller waves.  When $p = 2,3$ the system is completely integrable and one may construct 
\begin{quote}
{\it pure $n$-soliton solutions}, i.e. solutions whose difference from the linear superposition of $n$ independent solitary waves, each moving with its own characteristic speed, goes to zero in e.g. $H^1$ as time goes to either forward or backward infinity.  
\end{quote}
An important and long-standing observation about these solutions is that when a larger, faster soliton overtakes a smaller, slower soliton the collision is elastic - there is no energy lost to dispersion.  It has long been an open question as to what extent this phenomenon persists for non-integrable Hamiltonian systems.  When $p = 4,5$ the gKdV equation is believed to be non-integrable and thus the inverse scattering transform and other methods for constructing pure $n$-soliton solutions are unavailable.  However, one may (and Martel does \cite{martel:2005}) construct 
\begin{quote}
{\it asymptotic $n$-soliton states}, i.e. solutions whose long time behavior (in forward but not necessarily backward time) is precisely the same as that of the linear superposition of $n$ independent solitary waves, each moving with its own characteristic speed.  
\end{quote}
The construction of such solutions does not address the question of whether or not the collision is elastic, but it is a crucial step in the construction of 
\begin{quote}
{\it nearly $n$-soliton solutions}: Solutions whose long time behavior is precisely the same the the linear superposition of one collection of $n$ solitary waves as time goes backward to infinity and precisely the same (when measured in a norm which ignores dispersion) as the linear superposition of a different, and likely lower energy, collection of solitary waves as time goes forward to infinity.  
\end{quote}
The difference in the energy of the two asymptotic $n$-soliton states accounts for the energy which has been lost to dispersion as the solitons which comprise the nearly $n$-soliton solutions collide inelastically. 

One limitation of the gKdV equation is that it supports wave propagation in only one direction.
More realistic models, like the equations for water waves allow waves to propagate in any
direction.  Another class of equations which supports two-wave wave propagation is
the Fermi-Pasta-Ulam system of Hamiltonian differential equations:
\begin{equation}
\left\{ \ba{l}
\dot{r} = (S - I)p \\ \\
\dot{p} = (I - S^{-1})V'(r)
\ea \right.
\label{eq:FPU}
\end{equation}
where $u = (r,p) \in \ell^2 \times \ell^2$, $(Sx)_n = x_{n+1}$ is the shift on $\ell^2$, and $V : \R \to \R$ is an interaction potential which satisfies $V(0) = V'(0)  = 0$, $V''(0) > 0$, and $V'''(0) \ne 0$.  
(In fact, it is simple to show that one can assume, without loss of generality that $V''(0) =1$ and we
will assume this normalization from now on.) 
An extensive theory of small amplitude, long wavelength solitary waves has been developed for $\eqref{eq:FPU}$.  In particular, waves with amplitude $\sim \eps^2$, wavelength $\sim \eps^{-1}$, and wavespeed $\sim 1 + \eps^2$ exist \cite{friesecke:1994}, are exponentially localized and supersonic \cite{friesecke:1999}, and stable in a variety of senses of the word \cite{friesecke:2004, mizumachi:2007}.  We shall denote the solitary wave with speed $c = \pm(1 + \eps^2/24)$ by $u_c = (r_c,p_c)$. 
\begin{rmk}  Of course, any translate of the solitary wave $u_c$ is also a solitary wave of
the same speed.  For definiteness, we assume that $u_c$ is the member of this
family for which the profile $r_c$ has its (unique) maximum at zero.
\end{rmk}

We have recently investigated the long-time behavior of counterpropagating solitary waves in the small-amplitude long-wavelength regime of the FPU model \cite{hoffman:2008}.  In particular we have demonstrated that solutions exist which are close to the linear superposition of two solitary waves for all time and that moreover these solutions are asymptotically stable with respect to perturbations which are exponentially localized and orbitally stable with repsect to perturbations which live in $\ell^2$.  The purpose of this note is to prove the existence of asymptotic counterpropagating $2$-soliton states in the long-wavelength, low-amplitude regime of the FPU model.

\begin{thm} \label{thm:main}
The FPU model admits asymptotic two soliton states.  That is, given $c_-$ sufficiently close to and less than $-1$ and $c_+$ sufficiently close to and greater than $1$, there is an initial datum $\bar{u}_0 \in \ell^2$ such that the solution of $\eqref{eq:FPU}$ with initial data $u(0) = \bar{u}_0$ satisfies
\begin{equation}
\lim_{t \to \infty} \| u(t) - u_{c_+}(\cdot - c_+ t - \gamma_+(t)) - u_{c_-}(\cdot - c_- t - \gamma_-(t))\| = 0
\label{eq:thmest}
\end{equation}
for some functions $\gamma_\pm$ which satisfy $\lim_{t \to \infty} \dot{\gamma}(t) = 0$.
\end{thm}

\begin{rmk} Here, and in what follows, $\| \cdot \|$ refers to the $\ell^2(\Z)$ norm.
\end{rmk}

At the heart of this existence result is the following orbital stability result which was proven in \cite{hoffman:2008}.
\begin{thm} \label{thm:os}
There are positive constants $C$, $T_0$, $b$, $\delta_0$, and $\eps_0$ such that for any $\eps \in (0,\eps_0)$ the following hold.  Suppose that 
\[ \| u_0(\cdot) - u_{c_+}(\cdot - \tau_+^*) - u_{c_-}(\cdot - \tau_-^*) \| < \delta < \delta_0 \]
with $\tau_-^* < -T < T < \tau_+^*$ and $T > T_0$.  Then there are smooth real valued functions of a real variable $c_+$, $c_-$, $\tau_+$, and $\tau_-$ such that the solution $u$ of $\eqref{eq:FPU}$ with $u(0) = u_0$ satisfies
\begin{equation}
 \| u(t,\cdot) - u_{c_+(t)}(\cdot - \tau_+(t)) - u_{c_-(t)}(\cdot - \tau_-(t)) \| < C\eps^{-3/2}\delta + Ce^{-b \eps T}
 \label{eq:os1}
 \end{equation}
 for all $ t > 0$ 
Moreover, 
\begin{equation}
 |c_+(t) - c_+(t_0)| + |c_-(t) - c_-(t_0)| < C\eps^{-4}\delta +Ce^{-b \eps T}. 
 \label{eq:os2}
 \end{equation}

Finally, the limit $\lim_{t \to \infty} (c_+(t),c_-(t),\dot{\tau}_+(t),\dot{\tau}_-(t))$ exists with
\begin{equation}
| \dot{\tau}_\pm(t) - c_\pm(t)| \le C\eps^{-3/2}\delta + Ce^{-b \eps T}
\label{eq:os3}
\end{equation}
holding uniformly for $t > 0$ and $\lim_{t \to \infty} |\dot{\tau}_\pm(t) - c_\pm(t)| = 0$ 
\end{thm}

Theorem $\ref{thm:os}$ follows by combining the results of Theorem 2.2, Theorem 4.5, and Lemma 3.7 from \cite{hoffman:2008}.

\section{Preliminaries}
In this section we collect several lemmata which streamline the proof of Theorem $\ref{thm:main}$.  The first such lemma provides an estimate on the residual of the linear superposition of two solitary waves which is valid on exponentially long time scales.  The application of this lemma is a key step in the proof of Theorem $\ref{thm:main}$.  The other lemmas in this section serve to verify the smallness condition $\eqref{eq:gsmall}$ in the hypotheses of Lemma $\ref{lem:dwe}$.

\begin{lemma} \label{lem:dwe}
Consider the discrete wave equation which is driven, from rest, by small inhomogeneous and nonlinear terms:
\begin{equation}
 \left\{ \ba{l} \dot{\rho} = (S - I)\pi \\ \\ \dot{\pi} = (I - S^{-1})\rho + g(\pi,\rho,t) \\ \\
\rho(0) = \pi(0) = 0
\ea \right.
\label{eq:dwe}
\end{equation}

Let $N$ and $c_\pm$ be given real numbers.  Define $N_\pm := c_\pm N$.  Let $\langle x, y \rangle := \sum_{j \in \Z} x_j y_j$ denote the usual $\ell^2$ inner product and define weighted inner products by 
$\langle x, y \rangle_\pm := \sum_{j \in \Z} e^{\pm 2\eps^{3/2}(j - (N_\pm - c_\pm t))} x_j y_j = e^{\mp 2 \eps^{3/2}(N_\pm-c_\pm t)} \sum_{j \in \Z} e^{\pm 2\eps^{3/2}j} x_j y_j$.  Define the ``energy'' of the solution $(\rho,\pi)$ of $\eqref{eq:dwe}$ by 
\[
{\bf E}(t) := \frac{1}{2}\left(\|\rho(t)\|^2 + \|\rho(t)\|_+^2 + \|\rho(t)\|_-^2 + \|\pi(t)\|^2 + \|\pi(t)\|_+^2 + \|\pi(t)\|_-^2\right)
\]

Given positive constants $C$, $\tilde{C}$, $\eta$, and $b$, there is a positive constant $\eps_0$ such that as long as $g$ satisfies the smallness condition
\begin{equation}
\sum_{\diamond \in \{\ell^2,+,-\}}\| g(\pi,\rho,t)\|_\diamond \|\pi(t)\|_\diamond \le C\left[ (\eps^{3/2} + {\bf E}(t)^{1/2}){\bf E}(t) + \eps^{4+\eta} e^{-b\eps(N - t)}\right] 
\label{eq:gsmall}
\end{equation}
for all $t < N$ and for some $\eps \in (0, \eps_0)$, then the energy ${\bf E}$ satisfies the estimate 
\begin{equation}
{\bf E}(t) \le \tilde{C} \eps^3 e^{- b \epsilon (N-t)} \qquad t \in [0,N]
\label{eq:lem1est}
\end{equation}
\end{lemma}

\begin{proof}
Compute
\[ \ba{lllr}
\dot{{\bf E}}(t) & = & c_+ \eps^{3/2}(\|\rho(t)\|_+^2 + \|\pi(t)\|_+^2) - c_- \eps^{3/2}( \|\rho(t)\|_-^2 + \|\pi(t)\|_-^2) & (i) \\ \\
& & + \langle \rho, (S - I) \pi \rangle + \langle \pi, (I - S^{-1})\rho \rangle & (ii)\\ \\
& & + \langle \rho, (S - I) \pi \rangle_+  + \langle \pi, (I - S^{-1})\rho \rangle_+ & (iii) \\ \\
& & + \langle \rho, (S - I) \pi \rangle_- + \langle \pi, (I - S^{-1})\rho \rangle_- & (iv) \\ \\
& & + \langle \pi, g \rangle + \langle \pi ,g \rangle_+ + \langle \pi, g \rangle_- & (v)
\ea
\]
Note that $(i)$ is bounded by $(c_+ + c_-)\eps^{3/2}{\bf E}(t)$ and $(ii)$ is identically zero.  Since the shift $S$ is not self-adjoint on the weighted spaces, we estimate more carefully
\[ |(iii)| = |\langle (S^* - S^{-1})\rho,\pi \rangle_+| = |(e^{-2\eps^{3/2}} - 1) \langle S^{-1}\rho,\pi \rangle_+| \le C_0\eps^{3/2} \| \rho \|_+ \| \pi \|_+ 
\]
and similarly $|(iv)| \le C_0 \eps^{3/2} \|\rho\|_- \|\pi\|_-$.
The estimate $|(v)| \le C\eps^{3/2} \mathbf{E} + C\eps^{4+\eta}  e^{-b\eps|N - t|}$ follows immediately from $\eqref{eq:gsmall}$ and the Cauchy-Schwartz inequality as long as ${\bf E}(t) < \tilde{C} \eps^3 $.  
Choose and fix $\tilde{C} > 0$ and define $T := \sup\{ t \in [0,N] \; | \; {\bf E}(t) < \tilde{C} \eps^3 \}$ so that  

\[ \left\{ 
\ba{l} 
\dot{{\bf E}} \le C_3 \eps^{3/2} {\bf E} + C\eps^{4+\eta}  e^{-b\eps(N - t)} \qquad t \in [0,T] \\ \\
{\bf E}(0) = 0 
\ea \right.
\]

Apply Gronwall's inequality to see that 
\[ \ba{lllr} {\bf E}(t) & \le & C\eps^{4 + \eta} \int_0^t e^{C_2\eps^{3/2}(t-s)} e^{-b \eps (N - s)} ds & t \in [0,T] \\ \\
 & \le &\frac{C \eps^{4 + \eta} }{b\eps - C_2 \eps^{3/2}} e^{C_2 \eps^{3/2} t - b\eps N} (e^{(b\eps - C_2\eps^{3/2})t} - 1) & t \in [0,T]\\ \\
& \le & \frac{C\eps^{3 + \eta} }{b - C_2 \eps^{1/2}} e^{-b\eps(N - t)} \le \tilde{C} \eps^3 e^{-b\eps(N - t)}
& t \in [0,T] \ea 
\]
where in the last inequality we have taken $\eps$ sufficiently small.  Thus {\it a posteriori} we may take $T = N$.  This establishes $\eqref{eq:lem1est}$ and thus completes the proof.
\end{proof}

The following simple lemma establishes that the weighted space with norm ${\bf E}$ embeds compactly into $\ell^2$.
\begin{lemma} \label{lem:cpt}
Let $a > 0$ be given.  Let $X \subset \ell^2$ enjoy the estimate 
\[ R(X) := \sup_{x \in X} \sum_{j \in \Z} x_j^2(1 + e^{aj} + e^{-aj}) < \infty.\]  Then $X$ is compact in $\ell^2$.
\end{lemma}

\begin{proof}
We show that for each $\eta > 0$ there is an $N = N(\eta)$ and a discrete set $D = D(\eta) = \{x_1, \cdots x_N\}$ such that the $\eta$-neighborhood of $D$ contains $X$.
Let $\eta > 0$ be given.  Choose $M$ so large so that for any $x \in X$, we have $\sum_{|j| > M} |x_j|^2 < \eta$.  Such a choice is possible because $R(X)$ is finite.  Now let $D$ be a discrete set such that the $\eta$-neighborhood about $D$ contains the ball of radius $R(X)$ in $\{ x \in \ell^2 \; | \; x_j = 0 \mbox{ for } |j| > M\} \cong \R^{2M+1}$.  Such a set exists because $\R^{2M + 1}$ is locally compact.  Note that the $\eta$-neighborhood of $D$ contains $X$ as desired.
\end{proof}

The following lemma establishes that a key cross term is exponentially small.  We make use of this lemma in verifying the smallness condition $\eqref{eq:gsmall}$ in the hypotheses of Lemma $\ref{lem:dwe}$.

\begin{lemma} \label{lem:small2}
Let $y_j(t) := r_{c_+}(j-c_+(N-t) )r_{c_-}'(j-c_-(N-t)  - \sigma)$ for some $\sigma \in (0,1)$.  Then there is a $b > 0$ which may be chosen independent of $\eps$ and $\sigma$  such that the estimate 
\begin{equation}
 \|y(t)\|_+ + \|y(t)\|_- + \|y(t)\| < C \eps^{9/2}  e^{-b\eps (N - t)} \qquad t < N 
\label{eq:small2}
 \end{equation}
is valid.
\end{lemma}

\begin{proof}
It is known from \cite{friesecke:1999} that the wave profile $r_c$ and its derivative $r_c'$ enjoy the exponential estimates
\[ |r_c(\xi)| < C\eps^2 e^{-\eps b_0 |\xi|} \qquad |r_c'(\xi)| < C\eps^3 e^{-\eps b_0 |\xi|} \]
for some positive constants $C$ and $b_0$ which may be chosen independently of $\eps$.  Thus 
\[ |y_j(t) | < C \eps^5 e^{- b_0 \eps |j-c_+(N-t)|} e^{-b_0 \eps |j - c_-(N-t) |} \]
There are three regimes to consider: (i) \; $j > c_+(N-t)$, (ii) $c_-(N-t) < j < c_+(N-t)$ and (iii) $j < c_-(N-t)$.  We first compute the $\ell^2$ norm:
\[
\ba{lll}
\sum_{j \in \Z} y_j(t)^2 & = & 
\displaystyle{
\sum_{j < c_-(N-t) } y_j(t)^2 + \sum_{c_-(N-t)  \le j \le c_+(N-t) } y_j(t)^2 + \sum_{j > c_+(N-t) } y_j(t)^2 
}
\\ \\
& \le & 
\displaystyle{
C\eps^{10}e^{-b_0\eps \left[(c_+ - c_-)(N-t) \right]}\sum_{k > 0} e^{-2b_0 \eps k} 
} \\ \\
& & 
\displaystyle{
+ C\eps^{10}e^{-b_0\eps\left[(c_+ - c_-)(N-t) \right]} \sum_{c_-(N-t) \le j \le c_+(N-t) } 1
}
\\ \\ 
& & + 
\displaystyle{
C\eps^{10}e^{-b_0\eps\left[(c_+ - c_-)(N-t) \right]} \sum_{l > 0} e^{-2 b_0 \eps l}
}
\\ \\
& \le & C\eps^{10}e^{-b_0\eps\left[(c_+ - c_-)(N-t) \right]}\left( (c_+ - c_-)(N-t)  + \sum_{j \in \Z} e^{-2b_0 \eps |j|}\right) \\ \\
& \le & C\eps^9 e^{-b_0 \eps (c_+ - c_-)(N-t)} 
\ea
\]
After the first inequality we have changed indices $k = c_-(N-t)  - j$ and $l = j - (c_+(N-t) )$.  In the last line we have used the fact that $\sup_{x > 0} xe^{-ax} < Ca^{-1}$ for any $a > 0$ and in particular for $a = b_0 \eps$.  We have also used the fact that for $a > 0$, $\sum_{k > 0} e^{-a k} < Ca^{-1}$ with $a = 2b_0 \eps$.  We now compute the $\| \cdot \|_+$ norm:
\[ \ba{lll} 
\|y(t) \|_+^2 & = & \eps^{10} e^{-2\eps^{3/2} c_+(N-t)} e^{-2b_0 \eps[(c_+ - c_-)(N-t) ]} \left( \sum_{k > 0} e^{2\eps^{3/2}(k + c_+(N-t) )}e^{-4b_0\eps k} \right. \\ \\
& & + 
\left. \sum_{l > 0} e^{2\eps^{3/2}(c_-(N-t) - l)} e^{-4b_0 \eps l} 
+ \sum_{j \in (c_-(N-t) , c_+(N-t) )} e^{2\eps^{3/2}j} \right) \\ \\
& \le & C\eps^{10} e^{-\frac{3 b_0 \eps}{2} [(c_+ - c_-)(N-t) ]}\left( \eps^{-1} +  [(c_+ - c_-)(N-t) ]\right) \\ \\
& \le & C\eps^9  e^{-b_0 \eps (c_+ - c_-)(N-t)}
\ea 
\]
Here we have used the fact that $\eps^{3/2}c_+ < \frac{b_0 \eps}{2}(c_+ - c_-)$ for small $\eps$ to ignore all of the powers of $\eps^{3/2}$ and essentially reduce this computation to that of the $\ell^2$ norm above.  The estimate for $\| y(t) \|_-^2$ is similar.  After taking square roots, we obtain $\eqref{eq:small2}$; this concludes the proof.

\end{proof}

The following lemma will be used to establish that the inhomogeneous terms which drive the evolution of the perturbation of two well-separated solitary waves are bounded by a cross term of
the form shown to exponentially small in Lemma \ref{lem:small2}.
We leave the proof as an exercise.
\begin{lemma} \label{lem:MA511}
Let $f : \R \to \R$ be $C^2$ and fix $0$.  Then the function $G : \R^2 \to \R^2$ given by
\begin{equation}
G(x,y) := \left\{ \ba{lr} \frac{f(x+y) - f(x) - f(y)}{xy} & x,y \ne 0 \\ \\
\frac{f'(y)}{y} & x = 0, y \ne 0 \\ \\
\frac{f'(x)}{x} & x \ne 0, y = 0 \\ \\
f''(0) & x = y = 0 \ea \right.
\label{eq:Gdef}
\end{equation}
is Lipschitz continuous in any bounded set in $\R^2$. 
\end{lemma}

\section{The proof of Theorem $\ref{thm:main}$}
We are now ready to prove Theorem $\ref{thm:main}$.

\begin{proof}[Proof of Theorem $\ref{thm:main}$]
In broad strokes, the proof proceeds as follows.  Ultimately we wish to have a 
pair of solitary waves $u_{c_+}$ and $u_{c_-}$ moving away from each other to the right and
left respectively.  However, for the first step in the proof we consider their evolution {\em toward}
each other from a widely separated initial state.  Note that because of the time
reversibility of the FPU equations, if the solution $u_{c_+}(\cdot,t)
= (r_{c_+}(\cdot,t),p_{c_+}(\cdot,t))$ is a solitary wave which propagates to the right
with speed $c_{+} >1$, then $U^{R}_{c_+}(\cdot,t)
= (r_{c_+}(\cdot,-t),-p_{c_+}(\cdot,-t))$ is a solitary wave propagating to the left with speed
$-c_+$.  Likewise we introduce a solution $U^{L}_{c_-}$ which is the ``time reversal''
of the solitary wave $u_{c_-}$ and represents a solitary
wave propagating to the right.  With this notation we consider a sequence of solutions of
the FPU equations, $\{ u^N(\cdot,t) \}$, with initial conditions
$$
u^N(\cdot,0) = U^R(\cdot - c_+ N,0)+U^L_{c_-}(\cdot-c_- N,0)\ .
$$
Let this initial data evolve for time $N - T_0$, at which point it may be written as the linear superposition of colliding waves which remain separated by a distance of $(c_+ - c_-)T_0$ plus an error term.  We prove that this error term is uniformly small in $N$ in a space which embeds compactly into $\ell^2$.  Thus the sequence $u^N(N - T_0)$ converges in $\ell^2$ along a subsequence $N_j$ to some limit.  Use reversibility to make the waves in this limiting initial datum move apart rather than together and call the result $\tilde{u}_0$; its forward evolution $u$ is the solution that we seek.  At this point we rely on Theorem $\ref{thm:os}$ which says that so long as initial data starts close to the sum of two well-separated solitary waves which are moving apart from each other, the evolution remains close to the sum of these two solitary waves for all time.  Moreover, the distance to the sum of two solitary waves is a constant multiple of the initial distance plus a term which is exponentially small in the separation.  

We now use the fact that the solutions $u^{N_j}(t)$ converge
to the solution $u(t)$ whose initial condition is $u_0$.  Since
$u^{N_j}(t)$ is equal to a pair of solitary waves when $t=N_j$,
we find a sequence of times $t_j$, tending toward infinity,
such that for $t> t_j$, $u(t)$ is more and more closely approximated
by a pair of solitary waves.  Allowing $t$ to go to infinity
we find that $u(t)$ converges to a pair of solitary waves
with {\em no} radiation.  We note that this general method
of using apriori bounds to derive solutions with specified
properties at $t \to \infty$ has been used in a dissipative
context in Hamel and Nadirashvili \cite{hamel:1999}, and in the gKdV
equation by Martel \cite{martel:2005}.

We now present the details.  Let $u^N(j,0) := U^R_{c_+}(j - c_+N) + U^L_{c_-}(j - c_-N )$ and let $u^N(\cdot,t)$ be the evolution of $u^N(\cdot,0)$ under the FPU flow.  Define the residual $(\rho^N,\pi^N)$ by $u^N(j,t) = U^R_{c_+}(j  - c_+N + c_+t) + U^L_{c_-}(j - c_-N + c_-t) + (\rho^N(j,t),\pi^N(j,t))$.

Suppress the superscript $N$.  The residual $(\rho,\pi)$ satisfies the equation
\[ \left\{ \ba{l} \dot{\rho} = (S - I)\pi  \\ \\ \dot{\pi} = (I - S^{-1})\left[ V'(r_{c_+} + r_{c_-} + \rho) - V'(r_{c_+}) + V'(r_{c_-}) \right] \\ \\
\rho(0) = \pi(0) = 0 \ea \right. \]
where $u_c = (r_c,p_c)$ and $r_{c_+}$ and $r_{c_-}$ are evaluated at $j - c_+N  + c_+t$ and $j - c_-N  + c_-t$ respectively .  

Rewrite the right hand side of the equation for $\dot{\pi}$ as
\[(1-S^{-1}) \rho + (1-S^{-1})( V'(r_{c_+} + r_{c_-} + \rho) - V'(r_{c_+}) - V'(r_{c_-})- \rho).\]
The quantity $V'(r_{c_+} + r_{c_-} + \rho) - V'(r_{c_+}) - V'(r_{c_-})- \rho$ will play the role
of $g(\pi,\rho,t)$ in Lemma \ref{lem:dwe}, and we begin by decomposing it as:
\[ \ba{lllr}
V'(r_{c_+} + r_{c_-} + \rho) - V'(r_{c_+}) - V'(r_{c_-})- \rho & =  & 
V'(r_{c_+} + r_{c_-} + \rho) - V'(r_{c_+} + r_{c_-}) - V''(r_{c_+} + r_{c_-})\rho  & (i) \\ \\
& &  + (V''(r_{c_+} + r_{c_-}) - 1)\rho & (ii) \\ \\
& & + 
 V'(r_{c_+} + r_{c_-}) - V'(r_{c_+}) - V'(r_{c_-}) & (iii) \\ \\
\ea \]
From Taylor's theorem with remainder  we see that $|(i)| < C\rho^2$.  So
\begin{equation}\label{eq:iest}
\| ( 1 - S^{-1} )(i) \|_{\diamond} \le C \| \rho^2 \|_{\diamond}
\end{equation}
for $\diamond \in \{\ell^2,+,-\}$.  In like fashion, we observe that
$V''(0) = 1$ that $|(ii)| < C|(r_{c_+} + r_{c_-})\rho|$ and hence
\begin{equation}\label{eq:iiest}
\| ( 1 - S^{-1} )(ii) \|_{\diamond} \le C \| (r_{c_+} + r_{c_-})\rho \|_{\diamond}
\le C (\|r_{c_+}\|_{\ell^{\infty}} + \| r_{c_-} \|_{\ell^{\infty}}) \| \rho \|_{\diamond}
\le C \eps^2 \| \rho \|_{\diamond}
\end{equation}
The estimates on term $(iii)$ are somewhat more subtle.
First use Lemma $\ref{lem:MA511}$ to conclude
\[(I-S^{-1})(iii) \le C(I-S^{-1})r_{c_+}r_{c_-}. \]
Next, note that 
\[\ba{lll} 
(I - S^{-1})r_{c_+}(\cdot - \xi_+)r_{c_-}(\cdot - \xi_-) & = &
r_{c_+}(\cdot - \xi_+)(I-S^{-1})r_{c_-})(\cdot - \xi_-) + (S^{-1}r_{c_-}(\cdot - \xi_-))(I-S^{-1})r_{c_+}(\cdot - \xi_+) \\ \\
 &  = & r_{c_+}(\cdot - \xi_+)r_{c_-}'(\cdot - \zeta_-) + r_{c_-}(\cdot - \xi_- - 1)r_{c_+}'(\cdot - \zeta_+).
\ea
\]
Here $\xi_\pm = c_\pm(N-t) $ and $\zeta_\pm \in (\xi_\pm -1,\xi_\pm)$; in the second line we have used the Mean Value Theorem.  It now follows from Lemma $\ref{lem:small2}$ that 
\[ \| (I-S^{-1}) r_{c_+}r_{c_-} \|_\diamond < C\eps^{9/2} e^{-b\eps(N-t)}, \qquad t < N,  \]
for any $\diamond \in \{\ell^2,+,-\}$.  

Multiplying  the above bounds by
$\| \pi \|_{\diamond}$ and then
summing over $\diamond \in \{\ell^2,+,-\}$, we see that $\eqref{eq:gsmall}$ is satisfied 
and thus if we fix some $T_0 > 0 $,
 $(\rho^N(N - T_0),\pi^N(N - T_0))$ is bounded uniformly in the space $\ell^2 \cap \ell^2_+ \cap \ell^2_-$.  We established in Lemma $\ref{lem:cpt}$ that $\ell^2 \cap \ell^2_+ \cap \ell^2_-$ embeds compactly in $\ell^2$; it follows that $(\rho^N(N - T_0),\pi^N(N - T_0))$ converges strongly in $\ell^2$ along a subsequence $N_j$ to some limit $(\tilde{\rho},\tilde{\pi})$.  
 
 Note that the limiting solution consists of a pair of solitary waves at positions $c_+ T_0$
 and $c_- T_0$, plus a small remainder.  We will assume that $T_0$ is large enough that
 the ``overlap'' of the tails of the waves remains small.  As described above we
 now consider this limiting solution as the initial condition for a solution in which the
 two solitary waves move apart from one another.  Once again, we make use of the
 reversibility of the FPU equation (and the definition of $U^R_{c_+}$ and $U^L_{c_-}$)
 and write
\[
\bar{u}_0(j) = u_{c_+}(j - c_+ T_0) + u_{c_-}(j - c_- T_0) + (\tilde{\rho}( j),-\tilde{\pi}( j))
\]
and let $u$ be the solution of $\eqref{eq:FPU}$ with $u(\cdot,0) = \bar{u}_0$.  Note that due to reversibility and continuous dependence on initial conditions $u(\cdot,t)$ is close to 
$\bar{u}^N(\cdot,N- T_0 - t)$, where $\bar{u}^{N_j}$
represents the ``time reversed'' analogue of $u^{N_j}$ - i.e.
we insert a minus sign in front of second component of the solution (the $p$-component)
to account for the reversal of time.  In particular, at time $t$,
$u$ will be close to a pair of ``outward'' propagating solitary waves located at 
$c_{\pm} (T_0+t)$.   

The argument proceeds by using the orbital stability result, Theorem $\ref{thm:os}$ to control the forward evolution of $u$ for all time.  Due to the presence of the $\eps^{-b\eps T}$ term on the right hand sides of $\eqref{eq:os1}$ and $\eqref{eq:os2}$, it is not useful to regard $\bar{u}_0$ as the initial condition for $u$.  Instead, we consider a sequence of times $T_k \to \infty$ and regard $u(\cdot,T_k)$ as the initial condition for $u$.  The exponentially small terms in $\eqref{eq:os1}$ and $\eqref{eq:os2}$ are now order $\eps^{-b \eps T_k}$ and upon letting $k \to \infty$ they vanish.  The remainder of the proof makes this argument precise.

Choose sequences $\{ \delta_k \}$ and 
$\{ T_k \}$ such that $\delta_k \to 0$ and $T_k \to \infty$.
Then by continuous dependence on initial conditions there exists $N_k > T_k$ such that 
\begin{equation}
\| u(\cdot,T_k) - \bar{u}^{N_k}(\cdot,N_k-T_k-T_0) \| < \delta_k
\end{equation}

Now write 
\begin{eqnarray}\nonumber
&& \bar{u}^{N_k}(\cdot, N_k-T_k) = u_{c_+}(\cdot + c_+(N_k-T_k-T_0) )
+ u_{c_-}(\cdot + c_-(N_k-T_k-T_0) )  \\  \nonumber && \qquad \qquad + (\rho^{N_k}(\cdot,N_k-T_k-T_0),
-\pi^{N_k}(\cdot,N_k-T_k-T_0) )
\end{eqnarray}

By Lemma $\ref{lem:dwe}$ we have
\begin{equation}
\| (\rho^{N_k}(\cdot,N_k-T_k-T_0),-\pi^{N_k}(\cdot,N_k-T_k-T_0) ) \| < C \eps^3 e^{-b \eps (T_k+T_0)}
\end{equation}

Thus
\begin{equation}
\| u(\cdot,T_k) - u_{c_+}(\cdot + c_+(N_k-T_k) - c_+ N_k)
+ u_{c_-}(\cdot + c_- (N_k-T_k) - c_-  N_k) \| < \delta_k + C \eps^3 e^{-b \eps T_k} 
:= \tilde{\delta}_k
\end{equation}

Now apply Theorem $\ref{thm:os}$ with $u_0= u(\cdot,T_k)$.  Then for all $t>T_k$ we have functions $c_\pm^k$ and $\tau_\pm^k$ such that 
\begin{equation}
\| u(t) - u_{c_+^k(t)}(\cdot - \tau_+^k(t) ) - u_{c_-^k(t)}(\cdot - \tau_-^k(t) ) \|
< C \eps^{-3/2} \tilde{\delta}_k + C e^{-\eps T_k}
\end{equation}
holds.  
It remains to replace the $k$-dependent modulation parameters $\tau_\pm^k$ and $c_\pm^k$ with $k$-independent quantities.

In light of $\eqref{eq:os2}$ and the fact that $c \mapsto u_c$ is smooth, we may replace $c^k_\pm(t)$ with $c_\pm$ at the cost of an additional $C\eps^{-4}\tilde{\delta}_k + Ce^{-\eps T_k}$ on the right hand side.  Now define $\gamma_\pm^k(t) := \tau_\pm^k(t) - c_\pm t$.  In light of $\eqref{eq:os2}$ and $\eqref{eq:os3}$ we see that 
$|\dot{\gamma}_\pm^k(t)| < C\eps^{-4}\delta_k + Ce^{-b\eps T_k}$.
In particular, $\dot{\gamma}_\pm^k \to 0$ uniformly as $k \to \infty$.  Taking $k \to \infty$ we find that $u$ converges to a pair of solitary waves.  This completes the proof.
\end{proof}

{\bf Acknowledgements:} This work was funded in part by the National Science Foundation under grants DMS-0603589 and DMS-0405724.

\bibliography{AsymptoticBib}{}

\end{document}